\documentclass[conference]{IEEEtran}
\usepackage[left=1.62cm,right=1.62cm,top=1.85cm,bottom=2.6cm]{geometry}

\usepackage{amsthm,amssymb,graphicx,multirow,amsmath,color,amsfonts}
\usepackage[update,prepend]{epstopdf}%
\usepackage[noadjust]{cite}
\usepackage{tikz}
\usetikzlibrary{arrows,calc}		
\usepackage{bbm} 
\usepackage{p dfpages}
\usepackage{tabulary}
\usepackage{multirow}
\usepackage{comment}
\usepackage{dsfont}
\usepackage{pgf,tikz}
\usepackage{paralist}
\usepackage{amsfonts}
\usepackage{setspace,lipsum}
\usepackage{url}

\usepackage{graphicx}
\usepackage{bigints}
\usepackage{mathtools}
\usepackage[english]{babel}
\usepackage[autostyle]{csquotes}
\usepackage{textcomp}
\usepackage{varwidth}
\usepackage{tcolorbox}
\usepackage{lipsum}
 
\let\OLDthebibliography\thebibliography
\renewcommand\thebibliography[1]{
  \OLDthebibliography{#1}
  \setlength{\parskip}{0pt}
  \setlength{\itemsep}{0pt plus 0.3ex}
}

\newcommand\numberthis{\addtocounter{equation}{1}\tag{\theequation}}

\usepackage{amssymb,amsmath,amsfonts,amsthm}
\usepackage{mathrsfs}
\usepackage{cite}
\usepackage{array}
\usepackage{tabularx}
\usepackage{color}
\usepackage{mathtools}
\usepackage{subfigure}
\usepackage{mathtools}
\usepackage{tikz,pgf}
\usetikzlibrary{calc,positioning,mindmap,trees,decorations.pathreplacing}
\usepackage{graphicx}
\usepackage{bbm}
\usepackage[utf8]{inputenc}
\usepackage{algorithm}

\renewcommand{\IEEEQED}{\IEEEQEDopen}

\newcommand\blfootnote[1]{%
  \begingroup
  \renewcommand\thefootnote{}\footnote{#1}%
  \addtocounter{footnote}{-1}%
  \endgroup
}

\def\nb0{{\mathbf{0}}}
\def\nb1{{\mathbf{1}}}

\newtheorem{lemma}{Lemma}
\newtheorem{thm}{Theorem}

\newtheorem{remark}{Remark}

\allowdisplaybreaks 

\usepackage{setspace}	

\newcommand{\subparagraph}{} 
\usepackage{titlesec}        

\begin{document}

\title{Age of Information with On-Off Service\vspace{-.35em}}
	\author{Ashirwad Sinha, Praful D. Mankar, Nikolaos Pappas, Harpreet S. Dhillon \vspace{-.4em}}
	\maketitle
	\thispagestyle{empty}
\pagestyle{empty}
\begin{abstract}
This paper considers a communication system where a source sends time-sensitive information to its destination. We assume that both arrival and service processes of the messages are memoryless and the source has a single server with no buffer. Besides, we consider that the service is interrupted by an independent random process, which we model using the {\rm On-Off} process. For this setup, we study the age of information for two queueing disciplines: 1) {\em non-preemptive}, where the messages arriving while the server is occupied are discarded, and 2) {\em preemptive}, where the in-service messages are replaced with newly arriving messages in the {\rm Off} states. For these disciplines, we derive closed-form expressions for the mean peak age and mean age. 
\end{abstract}
\IEEEpeerreviewmaketitle 
\blfootnote{ A. Sinha and P. D. Mankar  are with SPCRC, IIIT Hyderabad, India (Email: ashirwad.sinha@students.iiit.ac.in,   praful.mankar@iiit.ac.in). N. Pappas is with the Link\"{o}ping University, Sweden (Email:  nikolaos.pappas@liu.se). H. S. Dhillon is with Wireless@VT, Department of ECE, Virginia Tech, Blacksburg, VA (Email:  hdhillon@vt.edu). The work of H. S. Dhillon was supported by US NSF under grant CNS-1814477.
	}\vspace{-5mm}
\section{Introduction}\vspace{-1mm}
The sixth generation (6G) of mobile communication system is envisioned  to support diverse use cases requiring massive machine-type communication (MMTC) and/or ultra-reliable low-latency communication  (URLLC) \cite{Jiang,popovski2022perspective}. Many of  these use cases will include remote monitoring and/or actuation where the timeliness of the  information may be crucial. 
For example, the sensors in MMTC may transmit time-sensitive updates, such as obstacle detection in autonomous car driving or fault detection in production chain, to a central processing unit. In such cases,  maintaining freshness of updates received at the central unit is critical. 
For this, a recently introduced metric, called {\em age of information} (AoI), is useful for measuring the freshness of information received at the destination \cite{SanjitKaul_2012}. Because of its analytical tractability, the mean AoI has emerged as a key performance indicator for the real-time  MMTC \cite{Jiang}. Besides,  the distribution of AoI is also useful for characterizing  the performance of URLLC \cite{Abdel}.
However, in several scenarios, the update service process gets interrupted, causing the undesired increase in the age of updates observed at the destination. Such scenarios, where the service toggles between {\em On and Off states}, include 1)  a mobile user going in and out of outage, 2) resources sharing/scheduling, and  3) energy harvesting communication. The AoI under service interruptions caused by an external process has not received much attention yet, which is the main theme of this paper.

{\em Related works:} The authors of \cite{SanjitKaul_2012} introduced the AoI metric and derived its average for {\rm M/M/1}, {\rm M/D/1}, and {\rm D/M/1} queues under the first-come-first-serve (FCFS) discipline. Subsequently, the same authors derived average age for {\rm M/M/1} queue in \cite{Sanjit_2012a} with two types of last-come-first-serve (LCFS) disciplines: 1) LCFS without preemption, where a new arriving update replaces the stale update in queue, and 2) LCFS with preemption, where a new arriving update replaces the in-service update. 
In a majority of works so far, the consideration of memoryless inter-arrival and service times processes act as the main facilitator.
Besides, the average age has also been analyzed for a general arrival/service process. For example, \cite{Soljanin_2017} derived average age for {\rm M/G/1/1} queue with HARQ, \cite{Rajat} derived the mean age and mean peak age of {\rm G/G/1}, {\rm M/G/1}, and {\rm G/M/1} queues with FCFS and LCFS in-service preemption disciplines, and  \cite{Ulukus_2021} derived the mean age and its bounds for {\rm G/G/1/1} queue with and without preemption. 

Some recent works have  started focusing on the distribution of age. The authors of \cite{Champati_2018} studied the  age distribution for {\rm D/G/1} queue with FCFS discipline, whereas \cite{Yoshiaki_2019} derived a general formula for the stationary distribution of age that is applicable for a wide class of update systems and demonstrated its use for various queues with FCFS and preemptive/non-preemptive LCFS policies. On similar lines, \cite{Nikos_2021} derived a general formula for the age distribution under ergodic settings, which is then used to derive the generating function of age for discrete systems. The authors of \cite{Mohamed_TIT} derived the MGF of age for {\rm M/M/1} queue with and without preemption. {Besides, a new approach based on the idea of stochastic hybrid system is developed for characterizing the distributional properties of age in \cite{Roy_SHS,Mohamed_SHS}.}
Besides, significant work exists on the  age  characterization for a variety of system settings, such as mutiple source system \cite{Yates_2019}, HARQ based update systems \cite{Soljanin_2017}, energy harvesting based update system with a single source \cite{Mohamed_TIT} and with multiple sources \cite{Roy_SHS,Mohamed_SHS}, 
FCFS queues in tandem \cite{kamage}, and {\rm M/M/1/2} system with random packet deadlines \cite{Sastry_2018a},   {\rm M/G/1} queue with vacation server \cite{Natarajan}, etc. 
However, the AoI under externally interrupted service process remains unexplored, despite its relevance in many situations as aforementioned.  

{\em Contributions:} This paper considers a modified {\rm M/M/1/1} system whose server toggles between on and off states according to an independent random process that models the  service interruptions. We model these interruptions using an {\rm On-Off} process such that the {\rm On} and {\rm Off} states duration are independent and exponentially distributed.
For this setup, we derive the mean age and mean peak age for two disciplines: 1) non-preemptive, where the updates arriving when server is busy are dropped, and 2) preemptive, where the in-service updates are replaced with new arrivals during {\rm Off} states. For the limiting cases of these {\rm On} and {\rm Off} state parameters, the derived expressions for the mean of age and peak age  approach to their mean values for the $M/M/1/1$ queue with continuous service. 
\begin{figure}[b!]
\centering
\includegraphics[trim={3cm 2.5cm 10.5cm 7.7cm},clip,width=.35\textwidth]{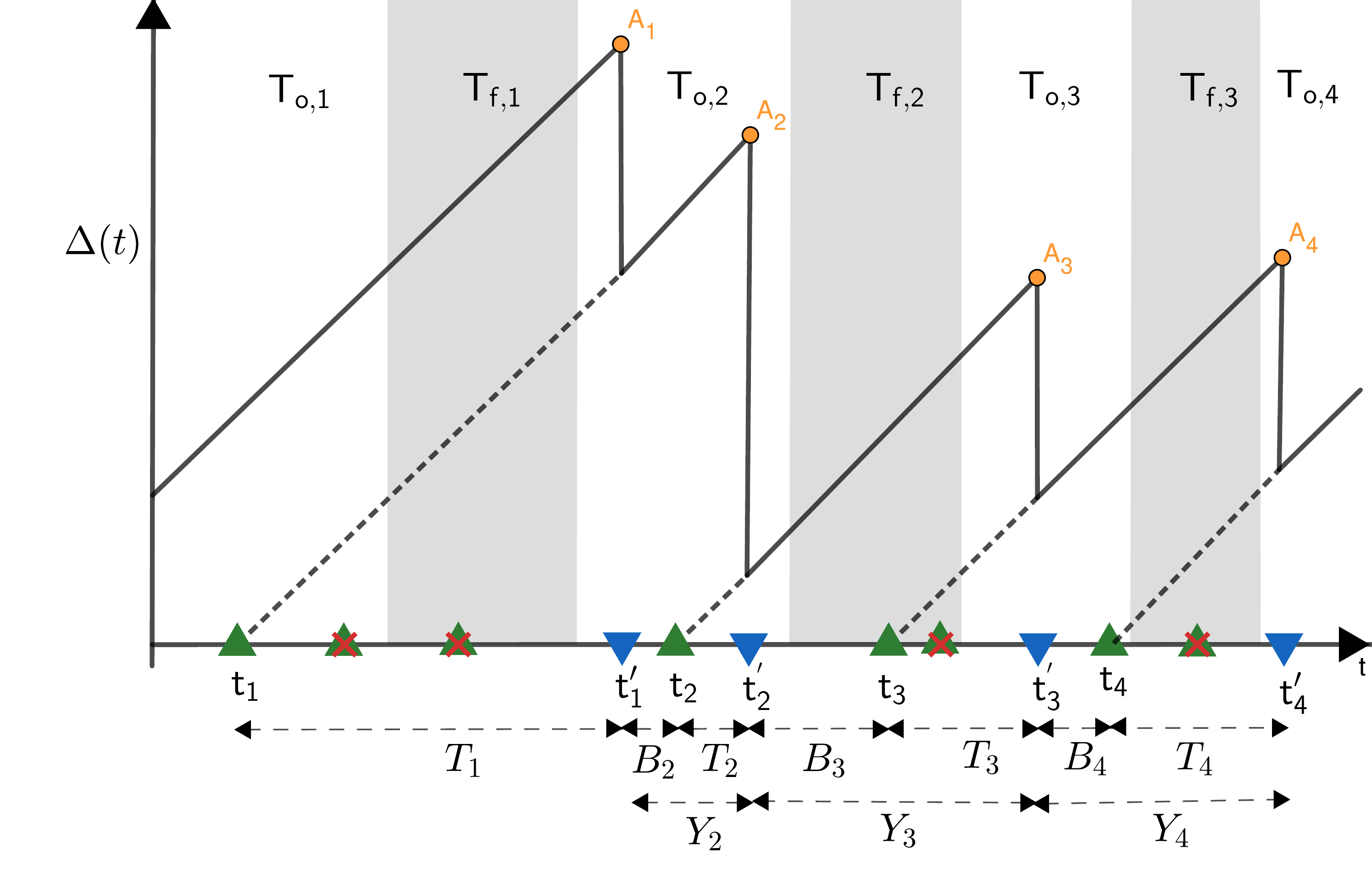} \\[2mm]\vspace{-3mm}
\includegraphics[trim={3cm 2.5cm 0.25cm 7.6cm},clip,width=.35\textwidth]{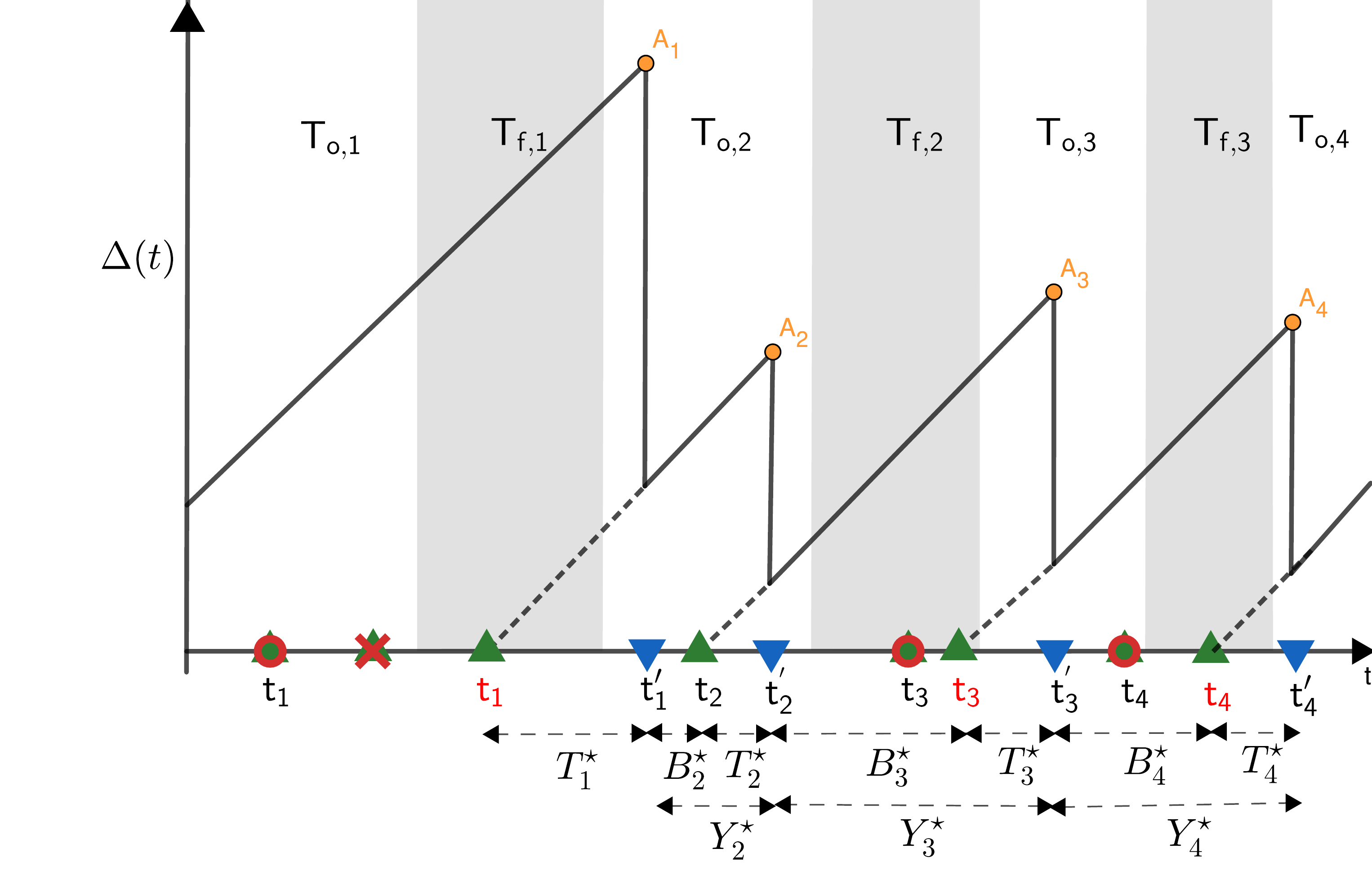} \vspace{-.3cm}
    \caption{Typical sample path of age $\Delta(t)$ for non-preemption (top) and preemption (bottom) policies. The green up and blue down arrow markers on $t$-axis indicate update arrivals and departures, respectively. The red cross and circle markers on $t$-axis show discarded and preempted updates, respectively.}\vspace{-3mm}
    \label{fig:age_sample_path}
\end{figure} \vspace{-4mm}
\section{System Model}\vspace{-2mm}
We consider a communication system where a source sends time sensitive updates about some physical process to its destination. It is assumed that the source has a single server with no packet storing facility and both the arrival and service processes follow exponential inter-arrival and service times with rates $\lambda$ and $\mu$, respectively. Further, we assume that the service process is interrupted by an {\rm On-Off} process, wherein the server operates normally during the {\rm On} states and stays idle, which we term as off, during the {\rm Off} states. The {\rm On} and {\rm Off} times are also assumed to be exponentially distributed with parameter $\kappa_{\rm o}$ and $\kappa_{\rm f}$, respectively. 
We consider two queueing disciplines: 1) {\em non-preemptive}, where the updates arriving while
the server is occupied are discarded, and 2) {\em preemptive}, where the in-service updates are replaced with newly arriving messages, if any, in the Off states. 
For such systems, we aim to analyze the AoI which is defined as $\Delta(t)=t-U(t)$, where $U(t)$ is the generation instance  of the most fresh update received by the destination. Fig. \ref{fig:age_sample_path} shows  sample paths of the age $\Delta(t)$ for non-preemptive and preemptive disciplines where $t_k$ and $t_k^\prime$ denote the arrival and departure instances of $k$-th delivered update. The service time of the $k$-th update is denoted as $T_k=t_k^\prime - t_k$, and the $i$-th {\rm On} and {\rm Off} states' periods are denoted as $T_{{\rm o},i}$ and $T_{{\rm f},i}$, respectively. Let $Y_k=t_k^\prime-t_{k-1}^\prime$  be the time between departures of $k$-th and $(k-1)$-th updates, and let $B_k=t_k-t_{k-1}^\prime$ denote the time required to arrive the $k$-th update  since $(k-1)$-th delivery. We denote $T_k$ and $Y_k$ as $T_k^\star$  and $Y_k^\star$ for the preemptive case. 

We focus  on the analysis of the {\em mean age} and {\em mean peak age}. The mean age, denoted by $\Delta$, is defined as the time mean of the age process $\Delta(t)$, whereas the mean peak age, denoted by $\bar{\mathcal{A}}$, is defined as the mean of age process $\Delta(t)$ observed just before the delivery of updates. For the non-preemptive discipline, these can be expressed as  
\begin{align}
    \bar{\mathcal{A}}&=\mathbb{E}[Y_k]+\mathbb{E}[T_{k-1}],
    \label{eq:mean_peak_age_defination}\\
    \text{and~}\Delta &= 0.5\lambda_e\mathbbm{E}[Y_k^2]+\lambda_e\mathbbm{E}[Y_kT_{k-1}],\label{eq:mean_age_defination}
\end{align}
respectively, where $\lambda_e$ is the effective arrival rate. Please refer to \cite[Section III]{Costa_2016} for more details. For the preemptive discipline, we can determine the mean peak age $\bar{\mathcal{A}}^\star$ and the mean age $\Delta^\star$ using $T_k^\star$ and $Y_k^\star$ in the above expressions. \vspace{-1mm}
\section{Age Analysis for ON-OFF service}\vspace{-1mm}
We first present the age analysis for the non-preemptive case.
From \eqref{eq:mean_peak_age_defination} and \eqref{eq:mean_age_defination}, and $Y_k=B_k+T_k$, it is clear that the mean age analysis requires the first two moments of service time $T_k$. 
As will be evident shortly, the key step in deriving the moments of $T_k$ is conditioning it on the arrival of the delivered update in {\rm On} and {\rm Off} states. The probability of this conditional event is presented in the following lemma.\vspace{-1mm}
\begin{lemma}
\label{lemma:prob_arrive_in_on_state}
The probability that the delivered update arrives in the {\rm On} state is 
\begin{equation}
{\rm P_{On}}={(\lambda+\kappa_{\rm f})}({\lambda+\kappa_{\rm o}+\kappa_{\rm f}})^{-1}.   
\label{eq:prob_arrive_in_on_state}
\end{equation}
\end{lemma}\vspace{-2mm}
\begin{proof}
Please refer to Appendix \ref{appendix:prob_arrive_in_on_state} for the proof.
\end{proof}\vspace{-1mm}
\begin{lemma}
\label{lemma:mean_Tk}
The first and second moments of service time $T_k$ for non-preemptive policy are
\begin{align*}
    {\rm \bar{T}_k^1}&=\frac{1}{\mu}+\frac{\kappa_{\rm o}}{\kappa_{\rm f}}\left[\frac{1}{\mu}+\frac{1}{\lambda+\kappa_{\rm o}+\kappa_{\rm f}}\right],\numberthis\label{eq:mean_Tk}\\
 \text{and}~{\rm \bar{T}_k^2}&=\left(\frac{1}{\mu+\kappa_{\rm o}}+\frac{1}{\kappa_{\rm f}}\right)^2 \left(1+3\frac{\kappa_{\rm o}}{\mu}+2\frac{\kappa_{\rm o}^2}{\mu^2}\right)+ \frac{\rm 1}{(\mu+\kappa_{\rm o})^2} \\
    &~~~+\frac{\mu+\kappa_{\rm o}}{\mu\kappa_{\rm f}^2}(1-2{\rm P_{On}})-\frac{2}{\mu\kappa_{\rm f}}{\rm P_{On}},
    \numberthis\label{eq:second_moment_Tk}
\end{align*}
respectively, where ${\rm P_{On}}$ is given in Lemma \ref{lemma:prob_arrive_in_on_state}.
\end{lemma}\vspace{-3mm}
\begin{proof}
Please refer to Appendix \ref{appendix:mean_Tk} for the proof.
\end{proof}\vspace{-1mm} 
The effective arrival rate for any given queueing system can be calculated using the its inter-departure times distribution $Y_k$ as $\lambda_{\rm eff}=\mathbb{E}[Y_k]^{-1}.$
Finally, using  the above lemmas  along with \eqref{eq:mean_peak_age_defination} and \eqref{eq:mean_age_defination}, we obtain the mean  age  in the following theorem.
\begin{thm}
\label{theorem:Peak_age_without_preemption}
For $M/M/1/1$ queue with {\rm On-Off} service under non-preemptive policy, the mean peak age and mean age   are 
\begin{align}
    \bar{\mathcal{A}} &=\frac{1}{\lambda}+\frac{2}{\mu}+\frac{2\kappa_{\rm o}}{\kappa_{\rm f}}\left[\frac{1}{\mu}+\frac{1}{\lambda+\kappa_{\rm o}+\kappa_{\rm f}}\right],\label{eq:mean_peak_age}\\
   \text{and~~} \Delta&=\frac{1}{1+\lambda {\rm \bar{T}_k^1}}\left[\frac{1}{\lambda}+\frac{\lambda}{2}{\rm \bar{T}_k^2}+{\rm \bar{T}_k^1}\right]+{\rm \bar{T}_k^1},\numberthis\label{eq:mean_age}
\end{align}
respectively, where ${\rm \bar{T}_k^1}$ and ${\rm \bar{T}_k^2}$ are given in Lemma \ref{lemma:mean_Tk}.
\end{thm}
\begin{proof}
The mean peak age given in \eqref{eq:mean_peak_age} directly follows by substituting  $\mathbb{E}[Y_{k}]=\mathbb{E}[B_k]+\mathbb{E}[T_k]=\frac{1}{\lambda}+{\rm \bar{T}_k^1}$  in \eqref{eq:mean_peak_age_defination}, where ${\rm \bar{T}_k^1}$ is given in Lemma \ref{lemma:mean_Tk}.
Next, using $Y_k=B_k+T_k$, \eqref{eq:mean_age_defination} and independence of $T_k$'s, the mean age can be written as 
\begin{align*}
\Delta=     \frac{1}{\mathbb{E}[Y_k]}\left[\mathbb{E}[T_k^2]+\mathbb{E}[B_k^2]+2\mathbb{E}[T_k]\mathbb{E}[B_k]\right] + \mathbb{E}[T_k].
\end{align*}
Further, substituting $\mathbb{E}[Y_k]=\frac{1}{\lambda}+{\rm \bar{T}_k^1}$,    and ${\rm \bar{T}_k^1}$ and $\mathbb{E}[T_k^2]={\rm \bar{T}_k^2}$ from Lemma \ref{lemma:mean_Tk}, provides \eqref{eq:mean_age}.
\end{proof} \vspace{-2mm}

Recall that in the preemptive policy  the in-service update  is replaced with a newly arriving update during {\rm Off} states. For the analysis of this case, the crucial step is to derive the mean of the service time $T_k^\star$. 
Similar to the analysis of  $T_k$, we derive the mean of $T_k^\star$ by conditioning on the arrival of delivered update in {\rm On}/{\rm Off} state. 
The probability of arrival of a delivered update in {\rm On} state is smaller, compared to  non-preemptive case, due to the replacement of older updates with new ones arrived in the {\rm Off} states. This probability is derived in Lemma \ref{lemma:prob_arrive_in_on_state_with_preemption}. \vspace{-2mm}
\begin{lemma}
\label{lemma:prob_arrive_in_on_state_with_preemption}
The probability that the successfully delivered update under preemption discipline arrives in the {\rm On} state is 
\begin{equation}
    {\rm P^\star_{On}} = {\rm P_{On}}(1-\beta)(1-\alpha\beta)^{-1},
    \label{eq:prob_arrive_in_on_state_with_preemption}
\end{equation}
where ${\rm P_{On}}$ is given in \eqref{eq:prob_arrive_in_on_state}, $\beta=\frac{\kappa_{\rm o}}{\mu+\kappa_{\rm o}}$ and $\alpha=\frac{\kappa_{\rm f}}{\lambda+\kappa_{\rm f}}$.
\end{lemma}
\begin{proof}
Please refer to Appendix \ref{appendix:prob_arrive_in_on_state_with_preemption} for the proof.
\end{proof}
\begin{lemma}
\label{lemma:mean_Tk_with_preemption}
The mean service time $T_k^\star$ with preemption is
\begin{equation*}
    {\rm \bar{T}_k^{1,\star}}=\frac{1}{1-\gamma}\left[\frac{1}{\lambda+\kappa_{\rm f}}+\frac{1}{\mu+\kappa_{\rm o}}\frac{\lambda+\kappa_{\rm o}+\kappa_{\rm f}-\mu}{\lambda+\kappa_{\rm o}+\kappa_{\rm f}}\right],\numberthis
    \label{eq:mean_Tk_with_preemption}
\end{equation*}
where $\gamma={\kappa_{\rm o}\kappa_{\rm f}}{(\lambda+\kappa_{\rm f})^{-1}(\mu+\kappa_{\rm o})^{-1}}$. 
\end{lemma}
\begin{proof}
Please refer to Appendix \ref{appendix:mean_Tk_with_preemption} for the proof.
\end{proof} 
Finally, using results from Lemma \ref{lemma:prob_arrive_in_on_state_with_preemption} and \ref{lemma:mean_Tk_with_preemption}, we obtain the mean age for the preemption case in the following theorem.\vspace{-1mm}
\begin{thm}
\label{theorem:Peak_age_with_preemption}
For $M/M/1/1$ queue with {\rm On-Off} service and preemption policy,  the mean peak age and mean age are 
\begin{align}
    \bar{\mathcal{A}}^\star =\frac{1}{\lambda}+\frac{1}{\mu}+\frac{\kappa_{\rm o}}{\kappa_{\rm f}}\left[\frac{1}{\mu}+\frac{1}{\lambda+\kappa_{\rm o}+\kappa_{\rm f}}\right] + {\rm \bar{T}_k^{1,\star}},\label{eq:mean_peak_age_preemption}\\
    \text{and~~}\Delta^\star=\frac{1}{1+\lambda {\rm \bar{T}_k^1}}\left[\frac{1}{\lambda}+\frac{\lambda}{2}{\rm \bar{T}_k^2}+{\rm \bar{T}_k^1}\right]+{\rm \bar{T}_k^{1,\star}},\numberthis\label{eq:mean_age_preemption}
\end{align}
respectively, where ${\rm \bar{T}_k^1}$ and ${\rm \bar{T}_k^2}$ are given in Lemma \ref{lemma:mean_Tk} and ${\rm \bar{T}_k^{1,\star}}$ is given in Lemma \ref{lemma:mean_Tk_with_preemption}.
\end{thm} \vspace{-2mm}
\begin{proof}
Due to preemption, the time of arrival of the update (that is getting delivered) since its last delivery is higher as compared to that under non-preemption case, i.e., $B_k^\star\geq B_k$. 
This increment is equal to the reduction in the service time due to preemption.  This is because the preemption under memoryless service process essentially replaces the older packet with the new one without affecting the service/inter-departure time statistics. Hence, we have $Y_k=T_k+B_k=T_k^\star+B_k^\star=Y_k^\star$.
This can also be verified from Fig. \ref{fig:age_sample_path}. This also implies that the effective arrival rates $\lambda_e$ with the preemption and  non-preemption disciplines are the same. Thus, using $\lambda_e=\frac{1}{\mathbb{E}[Y_k]}$, 
 \begin{align*}
     \mathbb{E}[Y_k^\star]&=\mathbb{E}[Y_k]=\lambda^{-1}+{\rm \bar{T}_k^1},\\
     \mathbb{E}[Y_k^{\star^2}]=\mathbb{E}[Y_k^2]&={\rm \bar{T}_k^2}+{2}{\lambda^{-2}}+2{\rm \bar{T}_k^1}\lambda^{-1},
 \end{align*}
   and following the steps provided in the proof of Theorem \ref{theorem:Peak_age_without_preemption}, we obtain expressions given in \eqref{eq:mean_peak_age_preemption} and \eqref{eq:mean_age_preemption}.
\end{proof}\vspace{-2mm}
\begin{remark}
The mean peak age and mean age derived for both non-preemptive  (in Theorem \ref{theorem:Peak_age_without_preemption}) and primitive  (in Theorem \ref{theorem:Peak_age_with_preemption}) disciplines
approach to $\frac{1}{\lambda}+\frac{2}{\mu}$ and $\frac{1}{\lambda}+\frac{2}{\mu}-\frac{1}{\lambda+\mu}$, respectively, as $\kappa_{\rm o}\to 0$ and/or $\kappa_{\rm f}\to\infty$. These limiting values are equal to the mean peak age and mean age observed under conventional $M/M/1/1$ queue \cite[Eq. (21) and Eq. (25)]{Costa_2016}. This is  expected since the service tends to appear as uninteruppted as the {\rm On} state duration becomes larger and/or the {\rm Off} state duration becomes smaller, in which case the considered queue discipline will behave similar to the conventional $M/M/1/1$ queue. 
\end{remark}\vspace{-2mm}
Fig. \ref{fig:peak_age} shows the age as a function of $\lambda$ for $\mu=1$ and different scales of the {\rm On-Off} state duration. 
The figure verifies that both peak and mean age are minimum for $M/M/1/1$ system without {\rm On-Off} service, which is expected. 
Besides, the figure also illustrates that the preemption policy provides smaller peak as well as mean age as compared to non-preemption policy, for any given configuration of parameters.
It can be seen that both the age metrics degrade 
 for $\kappa_{\rm o}> \kappa_{\rm f}$. 
 However, interestingly, it can be noted that the mean ages are smaller 
 for $\kappa_{\rm o}\geq \mu$.
 This is because of the smaller {\rm On-Off} cycles resulting in a frequent restart of the memoryless service such that it reduces the overall service time.  
\begin{figure}
    \centering
    \includegraphics[width=.43\textwidth]{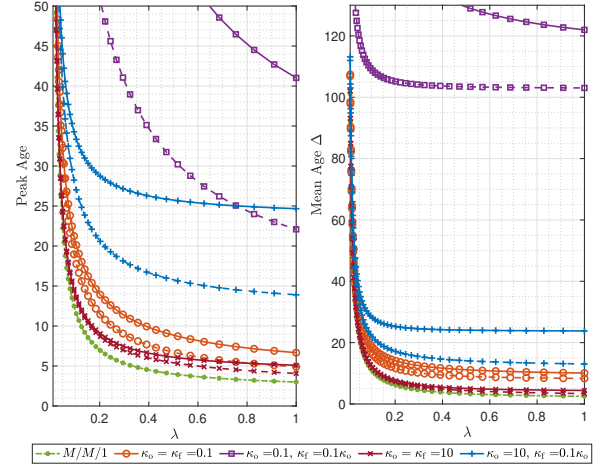} \vspace{-4mm}
    \caption{Age vs.  $\lambda$ for $\mu=1$, $\frac{\mu}{\kappa_{\rm o}}\in\{1,10\}$, and $\frac{\kappa_{\rm o}}{\kappa_{\rm f}}\in\{1,10\}$. The solid and dotted lines show the non-preemptive and preemptive policies, respectively. } 
    \label{fig:peak_age}
\end{figure} \vspace{-1mm}
\section{Summary} \vspace{-1mm}
This paper has analyzed the age of information for bufferless systems where both the update arrival and service processes are exponential and the service process is modulated with an independent {\rm On-Off} process. 
We have derived the mean peak age and mean age for the discipline where the updates arriving when the server is busy are discarded. Next, we extended the analysis for the preemptive discipline where the source is assumed to be capable of replacing the stale updates with fresh ones arrived during {\rm Off} states. The derived expressions for both the mean peak age and the mean age reduce to those for the simple $M/M/1$ queue as the {\rm On} state parameter $\kappa_{\rm o}\to 0$ and/or the {\rm Off} state parameter $\kappa_{\rm f}\to \infty$.\vspace{-1mm}
\appendix \vspace{-2mm}
\subsection{Proof of Lemma \ref{lemma:prob_arrive_in_on_state}}
\label{appendix:prob_arrive_in_on_state}\vspace{-1mm} Let $s_i$ and $s_i^\prime$ denote the start times of the $i$-th {\rm On} and {\rm Off} states, respectively.
Without loss of generality, we assume that a successful delivery occurs at $t=0$ and the corresponding {\rm On} state is the $1$-st {\rm On} state. The probability that next update 
 arrives in an {\rm On} state is equal to the probability that the next arrival occurs in the interval  $\bigcup_{i=1}^\infty (s_i,s_i^\prime]$. Note that $s_1=0$,  
    $s_i=\sum_{k=1}^{i-1}(T_{{\rm o},k}+T_{{\rm f},k})$ and $s_i^\prime =s_i + T_{{\rm o},i}$, for $i=2,3,\dots,$ where  $T_{{\rm o},k}$ and $T_{{\rm f},k}$ are the durations of $k$-th {\rm On} and {\rm Off} states, respectively.
Since these intervals are disjoint, the probability that next arrival occurs in {\rm On} state for given $(s_i,s_i^\prime)$'s is
\begin{align*}
    &{\rm P_{On}^{cond}}=\sum\nolimits_{i=1}^\infty F(s_i^\prime) -F(s_i),\\
    &\stackrel{(a)}{=}\sum\nolimits_{i=1}^\infty \exp(-\lambda s_i)-\exp(-\lambda (s_i + T_{{\rm o},i})),\\
    &=\sum\nolimits_{i=1}^\infty [1-\exp(-\lambda T_{{\rm o},i})]\prod\nolimits_{k=1}^{i-1}\exp(-\lambda (T_{{\rm o},k}+T_{{\rm f},k})),
\end{align*}
where $F(\cdot)$ is  the ${\rm CDF}$ of arrival time. Step (a) follows using $F(s)=1-\exp(-\lambda s)$. Now, deconditioning ${\rm P_{On}^{cond}}$ gives 
\begin{align*}
    {\rm P_{On}}&\stackrel{(a)}{=}\sum\nolimits_{i=1}^\infty \mathbb{E}_{T_{\rm o}}[1-e^{-\lambda T_{{\rm o}}}]\left(\mathbb{E}_{T_{\rm {o}}}\left[e^{-\lambda T_{{\rm o}}}\right]\mathbb{E}_{T_{\rm {f}}}\left[e^{-\lambda T_{{\rm f}}}\right]\right)^{i-1},\\
    &=\sum\nolimits_{l=0}^\infty \left[\frac{\kappa_{\rm o}}{\lambda+\kappa_{\rm o}}\frac{\kappa_{\rm f}}{\lambda+\kappa_{\rm f}}\right]^l\frac{\lambda}{\lambda+\kappa_{\rm o}},
\end{align*}
where Step (a) follows directly as both $T_{{\rm o},k}$s and $T_{{\rm f},k}$s are ${\rm i.i.d.}$ sequences. Finally,  by applying geometric series,  we obtain \eqref{eq:prob_arrive_in_on_state}. \vspace{-6mm}
\subsection{Proof of Lemma \ref{lemma:mean_Tk}}
\label{appendix:mean_Tk} \vspace{-2mm}
Let $Z$ denote the remaining service time of the update and $T_{\rm o}$ denote the duration of an {\rm On} state. 
Using the memoryless property of an exponential service, we can find the probability that the update  gets served in an {\rm On} state as 
\begin{align*}
    {\rm P_o}&=\mathbb{E}\left[\mathbb{P}[Z\leq t|T_{\rm o}=t]\right]=\frac{\mu}{\mu+\kappa_{\rm o}}.\numberthis\label{eq:prob_getting_servved_in_on_state}
\end{align*}
Let $Z^\prime$ denote the service time spent in an {\rm On} state in which the service gets completed. The distribution of $Z^\prime$ becomes
\begin{align*}
    f_{Z^\prime}(z)&=f_Z(z|Z<T_{\rm o})=\frac{\mathbb{P}(T_{\rm o}>z|Z=z)f_Z(z)}{\int_0^\infty \mathbb{P}(T_{\rm o}>Z|Z=z)f_Z(z){\rm d}z},\\
    &=(\mu+\kappa_{\rm o})\exp(-(\mu+\kappa_{\rm o})z).\numberthis\label{eq:service_time_in_ON_state}
\end{align*}
Let $T_{\rm o}^\prime$ denote the {\rm On} state duration in which the service does not get completed. We obtain  the distribution of $T_{\rm o}^\prime$ as
\begin{align*}
    f_{T_{\rm o}^\prime}(t)&=f_{T_{\rm o}}(t|Z>T_{\rm o})=(\mu+\kappa_{\rm o})\exp(-(\mu+\kappa_{\rm o})z).\numberthis\label{eq:ON_time_when_no_service}
\end{align*}
Let $G_o=0$ and $G_n=\sum_{l=1}^nT_{{\rm o},n}$. Consider $\mathcal{E}_n$ is an event where the service ends in the $n$-th {\rm On} state after its arrival. We can determine the probability of $\mathcal{E}_n$ as
\begin{align*}
    \mathbb{P}[\mathcal{E}_n]&= \mathbb{P}[G_n \geq Z>G_{n-1}],\\
    &\stackrel{(a)}{=}\mathbb{P}[G_n\geq Z]\mathbb{P}[Z>G_{l}]^{n-1}\stackrel{(b)}{=}{\rm P_o}\beta^{n-1},\numberthis\label{eq:prob_service_in_Nth_ON_state}
\end{align*}
where Step (a) follows from the memoryless  property of $Z$  and the fact that $T_{{\rm o},n}$'s are {\rm i.i.d.}s. Step (b) follows using \eqref{eq:prob_getting_servved_in_on_state} where $\beta=1-{\rm P_o}=\frac{\kappa_{\rm o}}{\mu+\kappa_{\rm o}}$.
Now, by the memoryless  property, we can directly consider independently distributed exponential service times spent in different {\rm On} states for deriving the expected service time by conditioning on the service times, {\rm On} state durations, and  $\mathcal{E}_n$.  Given the arrival in {\rm On} state, the mean of $T_k$ becomes ${\rm \bar{T}_k^{1,{\rm o}}}=$
\begin{align*}
    &\sum_{n=1}^\infty\left[\mathbb{E}[Z|Z<T_{{\rm o},n}]+\sum_{l=1}^{n-1}\mathbb{E}[T_{{\rm o},l}|Z>T_{{\rm o},l}]+\mathbb{E}[T_{{\rm f},l}]\right] \mathbb{P}[\mathcal{E}_n],\\
    &\stackrel{(a)}{=}\mathbb{E}[Z|Z<T_{{\rm o}}] + \sum_{n=1}^{\infty} (n-1)[\mathbb{E}[T_{{\rm o}}|Z>T_{{\rm o}}]+\mathbb{E}[T_{{\rm f}}]]{\rm P_o}\beta^{n-1},\\
    &\stackrel{(b)}{=}(\mu+\kappa_{\rm o})^{-1} + \left[(\mu+\kappa_{\rm o})^{-1}+\kappa_{\rm f}^{-1}\right]{\rm P_o}\mathcal{G}(n-1),\\
    &\stackrel{(c)}{=}\mu^{-1}+\kappa_{\rm o}(\mu\kappa_{\rm f})^{-1},\numberthis\label{eq:mean_Tk_condition_ON_state}
\end{align*}
  where $\mathcal{G}(f(n))=\sum_{n=1}^\infty f(n)\beta^{n-1}$,  Step (a) follows from \eqref{eq:prob_service_in_Nth_ON_state} and the fact that both $T_{{\rm o},l}$'s and $T_{{\rm f},l}$'s follow {\rm i.i.d.} distributions. Step (b) follows using \eqref{eq:service_time_in_ON_state}  and \eqref{eq:ON_time_when_no_service}. 
  Step (c) follows using $\mathcal{G}(n-1)=\beta(1-\beta)^{-2}$, $\beta=1-{\rm P_o}$, and some algebric simplifications.
 Similarly, we obtain the mean of $T_k$ conditioned on update arrived in an {\rm Off} state as ${\rm \bar{T}_k^{1,{\rm f}}}=$
\begin{align*}
    &\sum_{n=1}^\infty\left[\mathbb{E}[Z|Z<T_{{\rm o},n}]+(n-1)\mathbb{E}[T_{{\rm o}}|Z>T_{{\rm o}}]+n\mathbb{E}[T_{{\rm f}}]\right] \mathbb{P}[\mathcal{E}_n],\\
    &=\frac{1}{\mu+\kappa_{o}}+\frac{1}{\mu+\kappa_{\rm o}}{\rm P_o}\mathcal{G}(n-1)+\frac{1}{\kappa_{\rm f}}{\rm P_o} \mathcal{G}(n),\\
    &\stackrel{(a)}{=}\frac{1}{\mu+\kappa_{o}}+\frac{1}{\mu+\kappa_{\rm o}}\frac{\kappa_{\rm o}}{\mu}+\frac{1}{\kappa_{\rm f}}{\rm P_o}^{-1},\\
    &=\mu^{-1}+(\kappa_{\rm f}\mu)^{-1}(\mu+\kappa_{\rm o}),
\end{align*}
where Step (a) follows from the similar steps used in \eqref{eq:mean_Tk_condition_ON_state}. Thus, the mean of $T_k$  becomes ${\rm \bar{T}_k^1}={\rm \bar{T}_k^{1,o}}{\rm P_{On}}+{\rm \bar{T}_k^{1,f}}(1-{\rm P_{On}}).$

Next, substituting ${\rm P_{On}}$ from \eqref{eq:prob_arrive_in_on_state} and further solving gives \eqref{eq:mean_Tk}.
Now, we derive the second moment of $T_k$ using the similar approach presented above as
\begin{align*}
    {\rm \bar{T}_k^{2,o}}=\sum\nolimits_{n=1}^\infty\mathbb{E}\left[\left(Z^\prime+\sum\nolimits_{l=1}^{n-1}[T_{{\rm o},l}^\prime+T_{{\rm f},l}]\right)^2\right] \mathbb{P}[\mathcal{E}_n].
\end{align*}
Since $Z^\prime$ and $T_{{\rm o},l}^\prime$ are equal in distribution [see \eqref{eq:service_time_in_ON_state} and \eqref{eq:ON_time_when_no_service}], we can write
\begin{align*}
    {\rm \bar{T}_k^{2,o}}&=\sum\nolimits_{n=1}^\infty\mathbb{E}\left[\left(\sum\nolimits_{l=1}^{n}T_{{\rm o},l}^\prime+\sum\nolimits_{l=1}^{n-1}T_{{\rm f},l}\right)^2\right] \mathbb{P}[\mathcal{E}_n].
\end{align*}
Note that, since $T_{{\rm o},l}^\prime$ follows exponential distribution with parameter $(\mu+\kappa_{\rm o})$ independently, we have $S_{{\rm o},n}=\sum_{l=1}^nT_{{\rm o},l}^\prime\sim{\rm Gamma}(n,(\mu+\kappa_{\rm o})^{-1})$.  Similarly,   $W_{{\rm f},n-1}=\sum_{l=1}^{n-1}T_{{\rm f},l}\sim{\rm Gamma}(n-1,\kappa_{\rm f}^{-1})$. Using this, we can write ${\rm \bar{T}_k^{o,2}}$
\begin{align*}
    &=\sum\nolimits_{n=1}^\infty\left[\mathbb{E}[S_{{\rm o},n }^2]+\mathbb{E}[W_{{\rm f},n-1 }^2]+2\mathbb{E}[S_{{\rm o},n }]\mathbb{E}[W_{{\rm f},n-1 }]\right] \mathbb{P}[\mathcal{E}_n],\\
    &=\sum\nolimits_{n=1}^\infty\left[\frac{n(n+1)}{(\mu+\kappa_{\rm o})^2}+\frac{n(n-1)}{\kappa_{\rm f}^2}+2\frac{n}{\mu+\kappa_{\rm o}}\frac{n-1}{\kappa_{\rm f}}\right] \mathbb{P}[\mathcal{E}_n],\\
    &\stackrel{(a)}{=}\frac{{{\rm P_o} \mathcal{G}(n(n+1))}}{(\mu+\kappa_{\rm o})^2}+ \left[\frac{1}{\kappa_{\rm f}^2}+\frac{2}{\kappa_{\rm f}(\mu+\kappa_{\rm o})}\right]{\rm P_o}\mathcal{G}(n(n-1)),\\
    &={\rm P_o}\left(\frac{1}{\mu+\kappa_{\rm o}}+\frac{1}{\kappa_{\rm f}}\right)^2\mathcal{G}(n^2) \\
    &+ {\rm P_o}\left(\frac{1}{(\mu+\kappa_{\rm o})^2}-\frac{1}{\kappa_{\rm f}^2}-\frac{2}{\kappa_{\rm f}(\mu+\kappa_{\rm o})}\right)\mathcal{G}(n),\numberthis\label{eq:second_moment_Tk_condition_ON_state}
\end{align*}
where Step (a) follows using \eqref{eq:prob_service_in_Nth_ON_state}.
With some algebraic calculations, we obtain
$\mathcal{G}(n^2)=\mathcal{Z}P_o^{-1},$
where $\mathcal{Z}=1+3\frac{\kappa_{\rm o}}{\mu}+2\frac{\kappa_{\rm o}^2}{\mu^2}$. 
We also have $\mathcal{G}(n)={(1-\beta)^{-2}}={{\rm P_o}^{-2}}$.
Thus \eqref{eq:second_moment_Tk_condition_ON_state} becomes
\begin{align*}
    {\rm \bar{T}_k^{2,o}}&=\mathcal{Z}\left(\frac{1}{\mu+\kappa_{\rm o}}+\frac{1}{\kappa_{\rm f}}\right)^2+ \frac{\rm P_o}{\mu^2}-\frac{1}{{\rm P_o}\kappa_{\rm f}^2} -\frac{2}{\mu\kappa_{\rm f}}.\numberthis\label{eq:second_moment_Tk_condition_ON_state_1}
\end{align*}
Similarly, we determine the second moment of  $T_k$  conditioned on the service started in {\rm Off} state as
\begin{align*}
    &{\rm \bar{T}_k^{2,f}}\stackrel{(a)}{=}\sum\nolimits_{n=1}^\infty\mathbb{E}\left[\left(S_{{\rm o},n}+W_{{\rm f},n}]\right)^2\right] \mathbb{P}[\mathcal{E}_n],\\
    &=\sum\nolimits_{n=1}^\infty\left[\mathbb{E}[S_{{\rm o},n}^2]+\mathbb{E}[W_{{\rm f},n}]+2\mathbb{E}[S_{{\rm o},n}]\mathbb{E}[W_{{\rm f},n}]\right] \mathbb{P}[\mathcal{E}_n],\\
    &=\left[\frac{1}{(\mu+\kappa_{\rm o})^2}+\frac{1}{\kappa_{\rm f}^2}\right]{\rm P_o}\mathcal{G}(n(n+1))+\frac{2{\rm P_o}\mathcal{G}(n^2)}{\kappa_{\rm f}(\mu+\kappa_{\rm o})},\\
    &=\mathcal{Z}\left(\frac{1}{\mu+\kappa_{\rm o}}+\frac{1}{\kappa_{\rm f}}\right)^2 +\frac{\rm P_o}{\mu^2}+\frac{1}{{\rm P_o}\kappa_{\rm f}^2},\numberthis\label{eq:second_moment_Tk_condition_OFF_state_1}
\end{align*}
where Step (a) follows from $S_{{\rm o},n}=Z^\prime+S_{{\rm o},n-1}$.
Now, we can obtain the second moment of $T_k$ as $ {\rm \bar{T}_k^2}= {\rm \bar{T}_k^{2,o}}{\rm P_{On}}+{\rm \bar{T}_k^{2,f}}(1-{\rm P_{On}}).$
Finally, substituting \eqref{eq:second_moment_Tk_condition_ON_state_1} and   \eqref{eq:second_moment_Tk_condition_OFF_state_1} gives \eqref{eq:second_moment_Tk}.  \vspace{-1mm}
\subsection{Proof of Lemma \ref{lemma:prob_arrive_in_on_state_with_preemption}} \vspace{-1mm}
\label{appendix:prob_arrive_in_on_state_with_preemption}
The update arrived in an {\rm On} state will be delivered if it does not get preempted during the {\rm Off} states occurring in its service time. 
Thus, the probability that the update arrived in {\rm On} state will be delivered is 
\begin{equation*}
    {\rm P^\star_{\rm On}}={\rm P_{On}}{\rm P_{pre}^c},
    \numberthis\label{eq:prob_arrive_in_on_state_with_preemption_1}
\end{equation*}
where ${\rm P_{On}}$ is given in Lemma \ref{lemma:prob_arrive_in_on_state} and ${\rm P_{pre}^c}$ is the probability that the update will not get preempted.   By conditioning on the event $\mathcal{E}_n$ (defined in Appendix \ref{appendix:mean_Tk}), we  obtain 
\begin{align*}
    {\rm P_{pre}^c}&=\sum\nolimits_{n=1}^{\infty}\left[\prod\nolimits_{l=1}^{n-1}\mathbb{P}[{\rm No~arrival~in~} T_{{\rm f},l}]\right]\mathbb{P}[\mathcal{E}_n],\\
    &\stackrel{(a)}{=}{\rm P_o}\sum\nolimits_{n=1}^\infty \alpha^{n-1}\beta^{n-1}=(1-\beta)\frac{1}{1-\alpha\beta},
\end{align*}
where Step (a) follows using \eqref{eq:prob_service_in_Nth_ON_state},  $ \mathbb{P}[{\rm No~arrival~in~} T_{{\rm f},l}]=\alpha=\frac{\kappa_{\rm f}}{\lambda+\kappa_{\rm f}}$ and independence of $T_{{\rm f},l}$'s. Finally, substituting the above ${\rm P_{pre}^c}$ and  ${\rm P_{On}}$ from Lemma \ref{lemma:prob_arrive_in_on_state} in \eqref{eq:prob_arrive_in_on_state_with_preemption_1} completes the proof.  \vspace{-2mm}
\subsection{Proof of Lemma \ref{lemma:mean_Tk_with_preemption}} \vspace{-2mm}
\label{appendix:mean_Tk_with_preemption}
 We construct the proof on the similar lines as in Appendix \ref{appendix:mean_Tk}.
 Recall, in Appendix \ref{appendix:mean_Tk}, the probability $\mathbb{P}[\mathcal{E}_n]$  that the service of an update ends in the $n$-th {\rm On} state after its arrival is derived for the case of no preemption. However, under preemption, this probability depends on the probability that service takes $n$ {\rm On} states and there is no preemption during the {\rm Off} periods occurring during the service time. 
   Let $\mathcal{E}_n^\star$ denote the event where service of an update ends in the $n$-th {\rm On} state after its arrival under preemption.
    Note that $\mathcal{E}_n^\star$ includes $n$ {\rm i.i.d.} {\rm Off} periods in the service if it starts from an {\rm Off}  state, otherwise it includes $n-1$ {\rm i.i.d.} {\rm Off} periods.
    Thus, 
\begin{align*}
    \mathbb{P}[\mathcal{E}_n^\star]=\begin{cases}\frac{\mathbb{P}[\mathcal{E}_n]{\rm P_{n-1}^{no-pre}}}{\sum_{n=1}^\infty \mathbb{P}[\mathcal{E}_n]{\rm P_{n-1}^{no-pre}}}, &~\text{if ser. starts in {\rm On} st.}, \\
    \frac{ \mathbb{P}[\mathcal{E}_n]{\rm P_n^{no-pre}}}{\sum_{n=1}^\infty \mathbb{P}[\mathcal{E}_n]{\rm P_n^{no-pre}}}, &~\text{otherwise}, 
    \end{cases}
    \numberthis\label{eq:prob_en_star}
\end{align*}
 where ${\rm P_n^{no-pre}}$ is the probability that there is no preemption in the $n$ {\rm Off} states occurring during the service time. 
 We can directly evaluate ${\rm P_n^{no-pre}}$ as ${\rm P_n^{no-pre}}=\alpha^n$,
 where $\alpha$ is the probability that there is no preemption in a typical {\rm Off} state which is obtained in Appendix \ref{appendix:prob_arrive_in_on_state_with_preemption} as $\alpha=\frac{\kappa_{\rm f}}{\lambda+\kappa_{\rm f}}$.
 Thus, substituting $\mathbb{P}[\mathcal{E}_n]={\rm P_o\beta^{n-1}}$  and ${\rm P_{no-pre}}=\alpha^n$ in \eqref{eq:prob_en_star} and simplifying, we get
 \begin{equation}
     \mathbb{P}[\mathcal{E}_n^\star]=(1-\gamma)\gamma^{n-1},
     \label{eq:prob_service_in_Nth_ON_state_with_preemption}
 \end{equation}
 where $\gamma=\alpha\beta$,
 irrespective of whether the services started in the {\rm On} or {\rm Off} state.
 Note that the distribution of {\rm Off} state duration conditioned on no preemption is 
 \begin{align*}
     f_{T_{\rm f}^\prime}(t)=(\lambda+\kappa_{\rm f})\exp(-(\lambda+\kappa_{\rm f})t).
 \end{align*}

Now, similar to \eqref{eq:mean_Tk_condition_ON_state},  we obtain the mean of $T_k^\star$ given the update arrived in the {\rm On} state as
\begin{align*}
    &{\rm \bar{T}_k^{1\star,o}}=\sum\nolimits_{n=1}^\infty\left[\mathbb{E}[Z']+(n-1)\mathbb{E}[T_{{\rm o}}^\prime]+(n-1)\mathbb{E}[T_{{\rm f}}^\prime]\right] \mathbb{P}[\mathcal{E}_n^\star],\\
    &\stackrel{(a)}{=}\mathbb{E}[Z^\prime] +[\mathbb{E}[T_{{\rm o}}^\prime]+\mathbb{E}[T_{{\rm f}}^\prime]] (1-\gamma)\gamma\sum\nolimits_{n=1}^{\infty} (n-1)\gamma^{n-2},\\
    &=\frac{1}{\mu+\kappa_{\rm o}} + \left[\frac{1}{\mu+\kappa_{\rm o}}+\frac{1}{\lambda+\kappa_{\rm f}}\right](1-\gamma)\gamma\frac{\rm d}{\rm d \gamma}\sum\nolimits_{l=0}^\infty \gamma^l,\\
    &= \left[(\mu+\kappa_{\rm o})^{-1} +\gamma(\lambda+\kappa_{\rm f})^{-1}\right](1-\gamma)^{-1},\numberthis\label{eq:mean_Tk_condition_ON_state_with_preemption}
\end{align*}
where Step (a) follows using \eqref{eq:prob_service_in_Nth_ON_state_with_preemption}.
Similarly, we obtain the mean of $T_k^\star$ given the update arrived in the {\rm Off} state as
\begin{align*}
    &{\rm \bar{T}_k^{1\star,f}}=\sum\nolimits_{n=1}^\infty\left[\mathbb{E}[Z']+(n-1)\mathbb{E}[T_{{\rm o}}^\prime]+n\mathbb{E}[T_{{\rm f}}^\prime]\right] \mathbb{P}[\mathcal{E}_n^\star],\\ 
    &=\mathbb{E}[Z^\prime] +\mathbb{E}[T_{{\rm o}}^\prime] (1-\gamma)\gamma\sum\limits_{n=1}^{\infty} (n-1)\gamma^{n-2}+\mathbb{E}[T_{{\rm f}}^\prime] \sum\limits_{n=1}^{\infty} n\gamma^{n-1},\\
    &= \left[\frac{1}{\mu+\kappa_{\rm o}}+\frac{1}{\lambda+\kappa_{\rm f}}\right]\frac{1}{1-\gamma}.\numberthis\label{eq:mean_Tk_condition_OFF_state_with_preemption}
\end{align*}
Next, we obtain mean $\bar{\rm T}_k^\star$ as
\begin{equation*}
    {\rm \bar{T}_k^{1\star}}={\rm \bar{T}_k^{1\star,o}}{\rm P_{On}^\star}+{\rm \bar{T}_k^{1\star,f}}(1-{\rm P_{On}^\star}).
\end{equation*}
Finally, by substituting  ${\rm \bar{T}_k^{1*,o}}$ from \eqref{eq:mean_Tk_condition_ON_state_with_preemption}, ${\rm \bar{T}_k^{1*,f}}$ from \eqref{eq:mean_Tk_condition_OFF_state_with_preemption}, and ${\rm P_{On}^\star}$ from  Lemma \ref{lemma:prob_arrive_in_on_state_with_preemption} in the above expression and further simplifying, we obtain \eqref{eq:mean_Tk_with_preemption}. This completes the proof. \vspace{-2mm}

\end{document}